\documentclass[journal]{IEEEtran}

\usepackage{amsmath}
\usepackage{amsfonts}
\usepackage{amssymb}
\usepackage{amsthm}
\usepackage{tabularx}
\usepackage{mathrsfs} 

\usepackage{url}
\usepackage[colorlinks]{hyperref}

\newtheorem{remark}{Remark}

\newtheorem{Proposition}{Proposition}

\newtheorem{corollary}{Corollary}

\usepackage{stfloats}
\usepackage{float}
\usepackage{graphicx}
\usepackage{xcolor}
\usepackage{subfigure}

\makeatletter
\def\blfootnote{\xdef\@thefnmark{}\@footnotetext}
\makeatother

\begin{document}
%\title{Fluid Antenna-based Integration of RSMA with Phase-Mismatched STAR-RIS}
\title{Phase-mismatched STAR-RIS with\\FAS-assisted RSMA Users}

\author{Farshad~Rostami~Ghadi,~\IEEEmembership{Member},~\textit{IEEE}, 
            Kai-Kit~Wong,~\IEEEmembership{Fellow},~\textit{IEEE},\\
            Masoud~Kaveh, 
	    F.~Javier~L\'{o}pez-Mart\'{i}nez,~\IEEEmembership{Senior~Member},~\textit{IEEE}, 
	    Yuanwei~Liu,~\IEEEmembership{Fellow},~\textit{IEEE},\\
	    Chan-Byoung Chae,~\IEEEmembership{Fellow},~\textit{IEEE}, and
	    Ross Murch,~\IEEEmembership{Fellow},~\textit{IEEE}
%\vspace{-2mm}
}
\maketitle	

\begin{abstract}
This paper considers communication between a base station (BS) to two users, each from one side of a simultaneously transmitting-reflecting reconfigurable intelligent surface (STAR-RIS) in the absence of a direct link. Rate-splitting multiple access (RSMA) strategy is employed and the STAR-RIS is subjected to phase errors. The users are equipped with a planar fluid antenna system (FAS) with position reconfigurability for spatial diversity. First, we derive the distribution of the equivalent channel gain at the FAS-equipped users, characterized by a $t$-distribution. We then obtain analytical expressions for the outage probability (OP) and average capacity (AC), with the latter obtained via a heuristic approach. Our findings highlight the potential of FAS to mitigate phase imperfections in STAR-RIS-assisted communications, significantly enhancing system performance compared to traditional antenna systems (TAS). Also, we quantify the impact of practical phase errors on system efficiency, emphasizing the importance of robust  strategies for next-generation wireless networks.
\end{abstract}

\begin{IEEEkeywords}
Fluid antenna system, simultaneously transmitting and reflecting, energy-splitting, rate-splitting multiple access, outage probability, average capacity.
\end{IEEEkeywords}

\maketitle
\blfootnote{The work of F. Rostami Ghadi is supported by the European Union's Horizon 2022 Research and Innovation Programme under Marie Sk\l odowska-Curie Grant No. 101107993. The work of K. K. Wong is supported by the Engineering and Physical Sciences Research Council (EPSRC) under Grant EP/W026813/1. The work of M. Kaveh  is supported in part by the Academy of Finland under Grants 345072 and 350464. The work of F. J. L\'opez-Mart\'inez is funded in part by Consejer\'{i}a de Universidad, Investigaci\'{o}n e Innovaci\'{o}n (Junta de Andaluc\'{i}a) through grant EMERGIA20-00297, and in part by MICIU/AEI/10.13039/50110001103 and FEDER/UE through grant PID2023-149975OB-I00 (COSTUME). The work of C. B. Chae is supported by the Institute for Information and Communication Technology Planning and Evaluation (IITP)/NRF grant funded by the Ministry of Science and ICT (MSIT), South Korea, under Grant RS-2024-00428780 and 2022R1A5A1027646. The work of R. Murch was supported by the Hong Kong Research Grants Council Area of Excellence Grant AoE/E-601/22-R.}
\blfootnote{\noindent F. Rostami Ghadi and F. J. L\'opez-Mart\'inez are with the Department of Signal Theory, Networking and Communications, Research Centre for Information and Communication Technologies (CITIC-UGR), University of Granada, 18071, Granada, Spain (e-mail: $\rm \{f.rostami, fjlm\}@ugr.es)$.}
\blfootnote{\noindent K. K. Wong is affiliated with the Department of Electronic and Electrical Engineering, University College London, Torrington Place, WC1E 7JE, United Kingdom and also affiliated with Yonsei Frontier Lab, Yonsei University, Seoul, Republic of Korea (e-mail: $\rm kai\text{-}kit.wong@ucl.ac.uk$).}
\blfootnote{\noindent M. Kaveh is with the Department of Information and Communication Engineering, Aalto University, Espoo, Finland (e-mail: $\rm masoud.kaveh@aalto.fi$).}
\blfootnote{\noindent Y. Liu is with the Department of Electrical and Electronic Engineering, The University of Hong Kong, Hong Kong SAR, China (e-mail: $\rm yuanwei@hku.hk)$.}
\blfootnote{\noindent C. B. Chae is with School of Integrated Technology, Yonsei University, Seoul, 03722, Republic of Korea. (e-mail: $\rm cbchae@yonsei.ac.kr)$.}
\blfootnote{\noindent R. Murch is with the Department of Electronic and Computer Engineering and Institute for Advanced Study (IAS), Hong Kong University of Science and Technology, Clear Water Bay, Hong Kong SAR, China (e-mail: $\rm eermurch@ust.hk)$.}

%\blfootnote{\noindent Copyright (c) 2015 IEEE. Personal use of this material is permitted. However, permission to use this material for any other purposes must be obtained from the IEEE by sending a request to pubs-permissions@ieee.org.} 
%\blfootnote{Manuscript received January 25, 2021; revised XXX. The review of this paper was coordinated by XXXX.} 
%\blfootnote{Digital Object Identifier 10.1109/XXX.2021.XXXXXXX}
%\IEEEpeerreviewmaketitle
\blfootnote{Corresponding author: Kai-Kit Wong.}

\section{Introduction}\label{sec-intro}
\IEEEPARstart{T}{he rapid} growth of wireless communication is driven by the increasing demand for ultra-fast and low-latency connectivity across applications such as smart cities, massive Internet of Things (IoT) deployments, and real-time remote collaboration \cite{Gua2021Enabling,Tariq2020,Andrews20246G,10054381}. As sixth-generation (6G) wireless networks are envisaged to meet these ambitious demands, new challenges arise, including efficient spectrum utilization, interference management in dense environments, and ensuring reliable communication in high-mobility scenarios. Additionally, 6G systems must balance adaptability, energy efficiency, and support for diverse use cases in hostile environments.

To address these challenges, several Industry Specification Groups (ISGs) have been established to guide the development of key 6G technologies. Specifically, the multiple access techniques (MAT) ISG explores novel multiuser access strategies \cite{ETSI_RSMA}, and discusses the possible standardization of rate-splitting multiple access (RSMA). From the perspective of information theory, RSMA is a highly efficient multiple access technique. It manages inter-user interference by splitting user messages into common and private parts and at each user, adopts interference cancellation to remove the interference caused by the common message. The idea dates back in the 1990s \cite{Rimoldi-1996} and has sparked renewed interest in recent years \cite{Mao-2022tut,Clerckx-2023jsac}.

At the same time, 6G is expected to encounter challenges in coverage as the operating frequency moves up. For this reason, reconfigurable intelligent surfaces (RISs) have emerged as an intelligent solution to overcome this issue \cite{RefEE,Renzo-2020,Tang-2020,RIS-SRE-2023,Ref1new}. RIS is composed of a massive number of unit cells, and by controlling the phases of the reflected signals from the unit cells, it provides directed reflections to intended receivers, so as to reestablish the communication link. Given the potential of RISs, the RIS ISG was established to drive standardization efforts for RIS in enabling 6G networks \cite{ETSI_RIS}. 

However, conventional RIS implementations face deployment challenges, particularly in controlling phase shifts accurately under real-world conditions. Imperfections in hardware and dynamic channel variations cause phase-shift configuration errors, thereby degrading system performance \cite{Palomares2023enabling}. There is also the issue of having limited coverage if RIS is one-sided. Consequently, an interesting evolution of RIS is the emergence of simultaneously transmitting and reflecting RIS (STAR-RIS), which can provide coverage on both sides of the surface \cite{Liu2021star,Mu2022Simultaneously}. STAR-RIS is ideally implemented in the form of a smart window that lets signals reflected and transmitted intelligently. Motivated by the above discussion, it is of utmost relevance to integrate STAR-RIS technologies with RSMA when phase errors are considered. Presumably, RSMA effectively manages interference for high spectral efficiency while STAR-RIS can extend the wireless coverage for energy efficiency.

While this vision is fascinating, there are concerns if RSMA can actually deliver the high spectral efficiency under practical settings because interference cancellation is not always perfect and error propagation exists. Furthermore, the phase errors at STAR-RIS might render its beamforming ineffective. As such, STAR-RIS and RSMA alone may not be enough and additional degree-of-freedom (DoF) in the physical layer will be highly beneficial. In this context, fluid antenna system (FAS) can be a key element easing the burden of STAR-RIS and RSMA. 

FAS is an emerging concept broadly representing the new forms of reconfigurable antenna for shape and position flexibility \cite{wong2020FAS,wong2022bruce,wu2024fluid}. Instead of having a single signal sample in space at a fixed position, FAS facilitates processing signal samples in a given spatial region with a small number of radio-frequency (RF) chains. Artificial intelligence (AI) techniques have been identified to be a powerful tool to enable FAS \cite{Wang-2024ai} and a recent article in \cite{New2024aTutorial} provides a thorough tutorial on the topic. In \cite{Lu-2025}, the authors elaborated on the definition of FAS from electromagnetic theory. Note that the emerging concepts such as movable antennas \cite{zhu2024historical}, flexible-position antennas \cite{10480333} and pinching antennas \cite{Yang-2025pa}, are all within the remit of FAS. FAS is a timely concept that benefits from the recent advances in liquid-based antennas \cite{huang2021liquid,shen2024design,Shamim-2025}, movable arrays \cite{basbug2017design}, metamaterial-based antennas \cite{johnson2015sidelobe,hoang2021computational,Liu-2025arxiv} and pixel-based antennas \cite{zhang2024pixel}. The works in \cite{shen2024design,Shamim-2025,Liu-2025arxiv,zhang2024pixel} demonstrated testbeds validating the performance of FAS.

Since the introduction of the concept by Wong {\em et al.}~in 2020 \cite{Wong-fas2020cl,wong2020fluid}, FAS has become a thriving research topic due to the excitement for a new DoF in the physical layer. For instance, attempts have been made to advance the channel modelling of FAS for accuracy as well as tractability \cite{Khammassi2023,Espinosa-2024}. Diversity order analysis for FAS can also be found in \cite{New2023fluid,Vega2023asimple}. In \cite{New2024information}, the diversity-multiplexing trade-off of a point-to-point system with FAS at both ends was derived. Channel estimation for FAS has also been investigated in \cite{xu2024channel,zhang2023successive,Xu-2025ce}, with \cite{10751774} highlighting the importance of oversampling in the process. Moreover, recent efforts have also seen FAS being applied in secrecy communications \cite{Tang-2023,Ghadi-2024dec} and integrated sensing and communications (ISAC) \cite{Wang-2024isac,zhou2024fasisac,Zou-2024}.

\subsection{State-of-the-Art}
Although it makes sense to combine FAS with STAR-RIS and RSMA, consideration of these technologies together has never been done before and is not well understood. With that said, there have been researches studying the synergy between FAS and RIS \cite{Rostami2024performance,Rostami2024secrecy,Wong2023fluid,Zhu2024fluid,Wang2024performance}, considering phase errors for STAR-RIS \cite{Papazafeiropoulos2022coverage,Xu2022correlated,Qian2024performance,Khel2024analytical,Khalid2023simultaneously,Liu2024star}, STAR-RIS with RSMA \cite{Dhok2022rate,Katwe2023improved,Maghrebi2024cooperative,Krishnan2025star}, and combining FAS with RSMA \cite{Rostami2024fluid}. These works are reviewed below.

Equipping a two-dimensional (2D) FAS at the user side, the authors in \cite{Rostami2024performance} studied the performance of an optimized RIS-aided communication system, where they derived the outage probability (OP) and delay outage rate (DOR) in compact analytical expressions. By extending \cite{Rostami2024performance} to the non-orthogonal multiple access (NOMA) scenario, \cite{Rostami2024secrecy} analyzed the secrecy metrics in a secure RIS-aided FAS communication system. In \cite{Wong2023fluid}, randomized RISs were utilized as artificial scatterers to enhance multipath propagation, enabling FAS to differentiate user signals more effectively in multiuser scenarios. Later in \cite{Zhu2024fluid}, a novel FAS-enabled joint transmit and receive index modulation scheme was proposed for RIS-assisted millimeter-wave (mmWave) communication systems, enhancing spectral efficiency with reduced hardware complexity and power consumption. Additionally, by assuming a multi-RIS-assisted FAS communication system under Nakagami-$m$ fading, the authors in \cite{Wang2024performance} derived the OP expression using copulas. 

Regarding the studies of STAR-RIS with phase errors, the authors in \cite{Papazafeiropoulos2022coverage} analyzed the coverage probability of a STAR-RIS-assisted massive multiple-input multiple-output (MIMO) communication system in the presence of phase errors. Then \cite{Xu2022correlated} considered a STAR-RIS-aided two-user downlink system, proposed a correlated transmission and reflection phase-shift model considering the NOMA scheme, and derived the OP and diversity order for three phase-shift configuration strategies. More recently, \cite{Qian2024performance} investigated the impact of electromagnetic interference and phase errors on STAR-RIS-assisted cell-free massive MIMO systems. The authors also devised a projected gradient descent algorithm and fractional power control methods to improve spectral efficiency and mitigate performance degradation. Additionally, \cite{Khel2024analytical} evaluated the OP and ergodic capacity of STAR-RIS-assisted terahertz communications under practical factors such as inter-user interference and phase errors from discrete phase shifters, illustrating that phase errors significantly impact system performance and highlighting trade-offs between energy-splitting (ES) and mode switching (MS) protocols. On the other hand, by considering phase errors and hardware impairments, \cite{Khalid2023simultaneously} studied the ergodic rate for STAR-RIS-assisted NOMA downlink transmission. Following \cite{Khalid2023simultaneously}, the authors in \cite{Liu2024star} assumed perfect channel state information (CSI) and discrete phase shifts with quantization errors in STAR-RIS-aided communications under Nakagami-$m$ fading, deriving closed-form expressions for the OP and ergodic rate, as well as providing approximations for special cases, including random and continuous phase shifts. 

Apart from these, synergizing RSMA with other technologies has also generated many useful results. For instance, \cite{Dhok2022rate} studied a STAR-RIS-aided RSMA system, analyzing spatially correlated Rician channels under ES and MS configurations, and deriving expressions for the OP and channel capacity as well as asymptotic OP for both infinite and finite block-length transmissions. Subsequently, \cite{Katwe2023improved} presented a coupled phase-shift STAR-RIS-aided uplink RSMA system to improve spectral efficiency, formulating a resource allocation problem for joint optimization of power, decoding order, and beamforming. Recently, with an active STAR-RIS, \cite{Maghrebi2024cooperative} explored its role in optimizing cooperative RSMA in a downlink network, maximizing the sum rate while meeting rate, hardware, and power constraints, with significant gains over baseline methods. Most recently, \cite{Krishnan2025star} considered STAR-RIS-assisted RSMA, studying its performance in terms of OP and sum-rate under both cases of perfect and imperfect CSI. Finally, FAS has been shown to be effective in improving the reliability of RSMA. In \cite{Rostami2024fluid}, the OP of a FAS-assisted RSMA system was derived, showing that FAS combined with RSMA significantly outperforms the traditional antenna system (TAS) and NOMA.

\subsection{Motivation and Contributions}
The unique advantages of FAS in enhancing adaptability and mitigating interference, STAR-RISs in extending coverage and improving energy efficiency, and RSMA in maximizing spectral efficiency, have motivated us to explore the integration of these technologies. While each of these technologies offers great benefits on its own, their combined performance, especially in the presence of practical impairments such as phase errors in STAR-RISs under the RSMA scheme, is not well understood. Motivated by this, this paper seeks to investigate key aspects of the integrated system, specifically: (i) \textit{how the incorporation of FAS affects the performance of STAR-RIS-assisted communication systems}, and (ii) \textit{the impact of phase errors in affecting the performance of FAS-aided STAR-RIS systems when combined with RSMA signaling.} 

By addressing these critical questions, our aim is to provide valuable insights into the practical implementation of these technologies. To address this, we propose a novel FAS-assisted  STAR-RIS communication system, in which a TAS transmitter adopts the RSMA signaling technique to simultaneously serve FAS-equipped users across both the reflection and transmission regions, while accounting for phase errors in the STAR-RIS. Specifically, our main contributions are as follows:
\begin{itemize}
\item \textbf{Phase error modeling and channel gain analysis:} We begin by incorporating phase errors into the model and derive both the cumulative distribution function (CDF) and probability density function (PDF) for the equivalent channel gain at the FAS-equipped users. These are characterized by the multivariate $t$-distribution, which provides a more realistic representation of the channel conditions compared to traditional models, capturing the impacts of phase errors in practical systems.

\item \textbf{Performance analysis:} Using the derived CDF and PDF, we present compact analytical expressions for the OP as well as the asymptotic OP. Also, we introduce a novel heuristic approach to approximate the expectation of the maximum of many correlated Gamma random variables (RVs). This approach enables us to derive the average capacity (AC) of the system by leveraging Jensen's inequality, providing an efficient and tractable method for analyzing the system performance under phase errors.

\item \textbf{Numerical simulations and performance evaluation:} Through extensive simulations, we show that deploying FAS at the users in STAR-RIS-aided RSMA communication systems significantly enhances system performance in terms of both OP and AC. Our results illustrate that FAS greatly outperforms TAS by mitigating interference and improving adaptability. Moreover, the results reveal that RSMA achieves a lower OP compared to the NOMA scheme in both phase error and ideal phase scenarios. Additionally, we investigate the impact of practical phase errors in the STAR-RIS and provide insights into how these imperfections affect the overall system performance. Our findings highlight the performance degradation caused by phase errors, shedding light on the practical trade-offs in realistic deployments of STAR-RIS.
\end{itemize}

\subsection{Organization and Notations}
The remainder of this paper is organized as follows. Section \ref{sec-sys} presents the system model, including the channel and signal models, as well as the characterization of signal-to-noise-and-interference ratio (SINR). Then Section \ref{sec-perf} provides the main results, where the CDF, PDF, OP, and AC are derived. Section \ref{sec-num} presents numerical results that validate the analytical findings. Finally, Section \ref{sec-con} concludes the paper.

We use boldface uppercase and lowercase letters for matrices and vectors, e.g. $\mathbf{X}$ and $\mathbf{x}$, respectively. The operators $(\cdot)^T$, $(\cdot)^{-1}$, $\left|\cdot\right|$, $\max\left\{\cdot\right\}$, and $\det\left(\cdot\right)$ denote the transpose, inverse, magnitude, maximum, and determinant, respectively.

\section{System Model}\label{sec-sys}
\subsection{Channel and Signal Model}
We consider a STAR-RIS-aided communication framework as illustrated in Fig.~\ref{fig_model}, in which a TAS-equipped base station (BS) serves two users $u\in\left\{\mathrm{r,t}\right\}$, located in the reflection and transmission regions, with the assistance of a STAR-RIS that is composed of $K$ passive elements. Moreover, we consider the ES protocol,\footnote{The results in this paper can be extended to MS and time-switching (TS) configurations, where MS involves dynamically allocating different numbers of active STAR-RIS elements for transmission and reflection based on the network requirements, while TS allocates distinct time blocks for each mode, enabling flexible switching between transmission and reflection phases \cite{Mu2022Simultaneously}. These approaches allow for optimized resource management, adapting to varying network conditions and user demands.} where all elements of the STAR-RIS are configured to operate simultaneously in both reflection and transmission modes. The total radiated energy is divided into two portions, with the adjustable reflection and transmission coefficients represented by $ \beta_{k, \mathrm{r}}$ and $ \beta_{k, \mathrm{t}}$, respectively. Due to the passive nature of the STAR-RIS, these coefficients satisfy the constraint $\beta_{k, \mathrm{r}}^2 + \beta_{k, \mathrm{t}}^2 \leq 1.$\footnote{In practical STAR-RIS implementations, the reflection and transmission phase shifts are often coupled due to hardware constraints, limiting independent control. However, recent advances in metasurface design have explored configurations that enable independent or partially independent phase-shift control through advanced tuning mechanisms. For instance, a design featuring six metallic layers, including the top and bottom patches, two ground layers, and two middle layers with a reflective-type phase-shift circuit and bias circuit, allows for independent control of the phases of the reflection and transmission coefficients \cite{Hong2024star}. Our analysis assumes independent phase shifts to provide a general and flexible performance evaluation, aligning with the idealized case and potential future hardware improvements. Nonetheless, the proposed framework can be extended to accommodate coupled phase-shift models, as discussed in \cite{Xu2022performance}, where the impact of phase correlation is analyzed.} To simplify the system design, it is further assumed that all elements of the STAR-RIS share the same amplitude coefficients, such that $ \beta_{k, \mathrm{r}} = \beta_{\mathrm{r}},  \beta_{k, \mathrm{t}} = \beta_{\mathrm{t}}, \forall k = 1, \dots, K$.  Moreover, we assume that the direct link between the BS and users $u$ is blocked due to obstacles \cite{Rostami2023analytical}. It is also assumed that both the reflecting user (denoted as $\mathrm{r}$) and transmitting user (denoted as $\mathrm{t}$) are equipped with planar FAS, whereas the BS is outfitted with a single antenna. Each FAS-equipped user $u$ features a grid structure consisting of $N_u^l$ ports, uniformly distributed over a linear space of length $W_u^l\lambda$ with the carrier wavelength $\lambda$ for $l\in\left\{1,2\right\}$, where the total number of ports for user $u$ is $N_u=N_u^1\times N_u^2$ and the total size of the surface is $W_u=\left(W_u^1\times W_u^2\right)\lambda^2$. To simplify the indexing of these 2D structures, we introduce a mapping function $\mathcal{F}\left(n_u\right)=\left(n_u^1,n_u^2\right)$, $n_u=\mathcal{F}^{-1}\left(n_u^1,n_u^2\right)$, which converts the 2D indices into a one-dimensional (1D) form such that $n_u\in\left\{1,\dots,N_u\right\}$ and $n_u^l\in\left\{1,\dots,N_u^l\right\}$. In the given scenario, where the fluid antenna ports can switch freely to any position and can be positioned arbitrarily close to each other, the resulting channels exhibit spatial correlation. Therefore, under conditions of rich scattering, the covariance between two arbitrary ports, $n_u=\mathcal{F}^{-1}\left(n_u^1,n_u^2\right)$ and $\tilde{n}_u=\mathcal{F}^{-1}\left(\tilde{n}_u^1,\tilde{n}_u^2\right)$, for users $u$ is described as \cite{New2024information}
\begin{align}
\varrho_{n_u,\tilde{n}_u}=\mathcal{J}_0\left(2\pi\sqrt{\left(\frac{n_u^1-\tilde{n}_u^1}{N_u^1-1}W_u^1\right)^2+\left(\frac{n_u^2-\tilde{n}_u^2}{N_u^2-1}W_u^2\right)^2}\right),
\end{align}
where $\mathcal{J}_0\left(\cdot\right)$ denotes the zero-order spherical Bessel function of the first kind. This is a classical result when the channels have infinitely many scatterers in the environment.

\begin{figure}[!t]
\centering
\includegraphics[width=1\columnwidth]{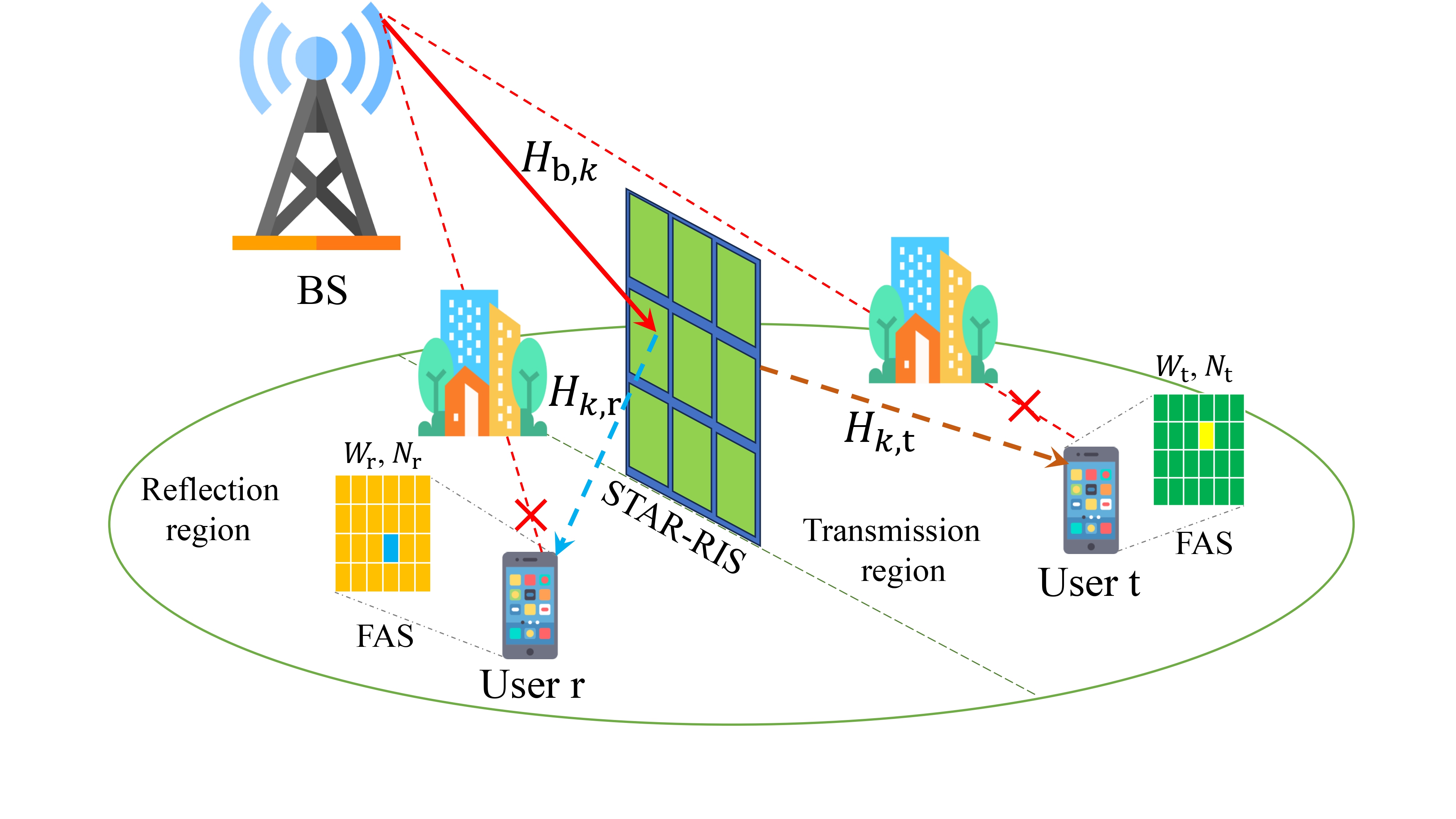}
\caption{A FAS-aided STAR-RIS system with two users, one on each side.}\label{fig_model}
\end{figure}

In this model, it is assumed that the BS uses a rate-splitting strategy to support communication with two users simultaneously. Each user's message $w_u$ is transmitted to its respective user $u$. According to the RSMA scheme, $w_u$ is divided into two components: a common message $w_{\mathrm{c}}$ and a private message  $w_{\mathrm{p},u}$. The common message  $w_{\mathrm{c}}$ is shared across both users and is encoded into a common stream $s_\mathrm{c}$ using a shared codebook accessible to both users. This common stream is designed to be decoded by every user. Meanwhile, the private messages  $w_{\mathrm{p},u}$ are independently encoded into individual private streams $s_{\mathrm{p},u}$, tailored for each user. Therefore, the resulting signal transmitted by the BS to each user is then expressed as
\begin{align}
	x=\underset{\text{common message}}{\underbrace{\sqrt{P_\mathrm{b}\alpha_\mathrm{c}}s_\mathrm{c}}}+\underset{\text{ private messages}}{\underbrace{\sqrt{P_\mathrm{b}\alpha_{\mathrm{p,r}}}s_{\mathrm{p,r}}+\sqrt{P_\mathrm{b}\alpha_{\mathrm{p,t}}}s_{\mathrm{p,t}}}},
\end{align}
where $P$ is the transmit power, while $\alpha_\mathrm{c}$ and $\alpha_{\mathrm{p},u}$ represent the power allocation factors for $s_\mathrm{c}$ and $s_\mathrm{p}$, respectively. These factors satisfy the constraint $\alpha_\mathrm{c}+\sum_{u\in\left\{\mathrm{r,t}\right\}}\alpha_{\mathrm{p},u}=1$, ensuring the total power is appropriately distributed among the streams. Hence, the received signal at user $u$ is defined as
\begin{align}\label{eq-y}
y_u = \sqrt{L_u}\sum_{k=1}^{K}H_{\mathrm{b},k}\beta_{k,u}\mathrm{e}^{j\phi_{k,u}}H_{k,u}x+z_u,
\end{align}
in which $z_u$ denotes the complex additive white Gaussian noise (AWGN) with zero mean and variance of $\sigma^2$, and $L_u=\left(d_\mathrm{SRIS}d_u\right)^{-\chi}$ represents the path-loss, where $\chi>2$ is the path-loss exponent and $d_i$ denotes the Euclidean distance between the BS placed at $\left(X_0,Y_0,Z_0\right)$ and the node $i\in\left\{{\rm SRIS}, {\rm r}, {\rm t}\right\}$ located at $\left(X_i,Y_i,Z_i\right)$, which can be determined as \cite{Rostami2024fluid}
\begin{align}
d_i = \sqrt{\left(X_0-X_i\right)^2+\left(Y_0-Y_i\right)^2+\left(Z_0-Z_i\right)^2}.
\end{align}
Besides, in \eqref{eq-y}, $H_{\mathrm{b},k}$ represents the fading channel coefficient between the BS and the $k$-th element of the STAR-RIS, while $H_{k,u}$ denotes the fading channel coefficient between the  $k$-th element of the STAR-RIS and user $u$. The fading coefficients are normalized to have unity power, and their corresponding average magnitudes for $k = 1, \dots, K$ are denoted as $a_\mathrm{b} = \mathbb{E}\left[\left|H_{\mathrm{b},k}\right|\right]$ and $a_u = \mathbb{E}\left[\left|H_{k,u}\right|\right]$. It is worth recognizing that the values $\{a_\mathrm{b}, a_u\}\leq 1$ under all circumstances, with equality occurring only in the limiting case of a deterministic fading channel, i.e., where fading is absent. For the
sake of compactness, $\tilde{a}_u=\sqrt{a_\mathrm{b}a_u}$ is defined \cite{Badiu2020Communication}. Additionally, $\phi_{k,u}$ defines the adjustable phases induced by the $k$-th element of the STAR-RIS during transmission and reflection. In this scenario, it can be demonstrated that the received signal-to-noise ratio (SNR) at the users is maximized when the phase shift for each element $\phi_{k,u}$ is controlled to co-phase the transmitting and reflecting links, i.e., $\phi_{k,u} = -(\eta_k+\zeta_{k,u})$, where $\eta_k$ and $\zeta_{k,u}$ represent the phases corresponding to the channel coefficients $H_{\mathrm{b},k}$ and $H_{k,u}$, respectively. However, in practice, the phases of the transmitting and reflecting elements cannot be set precisely due to inherent imperfections in phase estimation and the limited quantization of phase states at the STAR-RIS, which leads to a random phase error $\Theta_{k,u}$ \cite{Badiu2020Communication}, i.e., $\hat\phi_{k,u} = \phi_{k,u}+\Theta_{k,u}$. As such, the equivalent complex channel observed by user $u$ can be expressed as
\begin{align}
H_u = \frac{\beta_u}{K}\sum_{k=1}^{K}\left|H_{\mathrm{b},k}\right|\left|H_{k,u}\right|\mathrm{e}^{j\Theta_{k,u}}.
\end{align}
Accordingly, \eqref{eq-y} is rewritten as 
\begin{align}
y_u = \sqrt{L_u}KH_ux+z_u.
\end{align}
We assume that $\{\Theta_{k,u}\}$ are independent and identically distributed (i.i.d.) with common characteristic function expressed as a sequence of complex numbers $\left\{\varphi_p\right\}_{p\in \mathbb{Z}}$, i.e.,
\begin{align}
\varphi_p=\mathbb{E}\left[\mathrm{e}^{jp\Theta_{u}}\right], 
\end{align}
which are denoted as the $p$-th circular moments of $\Theta_{u}$. To account for phase errors induced by noisy estimates during the configuration stage, as well as due to channel aging, %To account for phase errors introduced by the limited number of phase shifts at the STAR-RIS, 
we adopt a phase error model based on the von Mises distribution \cite{Mardia2000directional}. Specifically, the phase estimation error, denoted as $\Theta_{k,u}$, is modeled as a zero-mean von Mises variable with a concentration parameter $\kappa$ that quantifies the accuracy of the phase estimation. The characteristic function is given by $\varphi_p=\frac{\mathcal{I}_p\left(\kappa\right)}{\mathcal{I}_p\left(\kappa\right)}$, where $\kappa$ is the concentration parameter capturing the accuracy of the estimation and $\mathcal{I}_p\left(\kappa\right)$ is the modified Bessel function of the first kind and order $p$. This framework effectively captures the impact of phase estimation inaccuracies and enables a precise analysis of their influence on system performance.

\subsection{SINR Characterization}
In RSMA, each user receives not only the common and its respective private message but the private messages for other users. This results in additional interference, complicating the process of decoding the desired messages. To manage this, RSMA adopts a two-step decoding strategy at each user. In the first step, the user focuses on decoding the common message while treating all private messages as noise. On the other hand, given by the FAS concept, it is assumed that only the port which maximizes the received SINR for the FAS-equipped users is activated. Hence, the received SINR for decoding the common message at the $n_u$-th port of user $u$ is defined as
\begin{align}
\gamma_{\mathrm{c},u} = \frac{\overline{\gamma}\alpha_\mathrm{c}L_uK^2g_{\mathrm{fas},u}}{\overline{\gamma}\left(\alpha_\mathrm{p,r}+\alpha_\mathrm{p,t}\right)L_uK^2g_{\mathrm{fas},u}+1},\label{eq-snr-c}
\end{align}
where $\overline{\gamma}=\frac{P_\mathrm{b}}{\sigma^2}$ denotes the average transmit SNR and $g_{\mathrm{fas},u}$ is the equivalent channel gain at user $u$ which is defined as 
\begin{align}\label{eq-gain}
g_{\mathrm{fas},u} = \max\left\{g_u^1,\dots,g_u^{n_u}\right\}=g_u^{n_u^*},
\end{align}
in which $g_u^{n_u}=\left|H_u^{n_u}\right|^2$ denotes the channel gain at the $n$-th port of user $u$ and $n_u^*$ denotes the best port index at user $u$ that maximizes the channel gain, i.e.,
\begin{align}
n_u^* = \arg\underset{n_u}{\max}\left\{\left|H_u^{n_u}\right|^2\right\}.
\end{align}

In the second stage, once the common message has been successfully decoded and subtracted, each user then decodes its private message while the private messages for another user are treated as residual interference. The SINR for decoding the private message at the $n_u$-th port of user $u$ is given by
\begin{align}
\gamma_{\mathrm{p},u} = \frac{\overline{\gamma}\alpha_{\mathrm{p},u}L_uK^2g_{\mathrm{fas},u}}{\overline{\gamma}\alpha_{\mathrm{p},\bar{u}}L_uK^2g_{\mathrm{fas},u}+1},\label{eq_SINR_p}
\end{align}
where $\bar{u}$ denotes the complement of $u\in\left\{\mathrm{r,t}\right\}$.

\section{Performance Analysis}\label{sec-perf}
\subsection{Statistical Characterization}
To analyze the system performance, we need to determine the statistical distribution of the equivalent channel gain $g_{\mathrm{fas},u}$, which is defined as the maximum of $N_u$ correlated RVs $g_u^{n_u}$, where $n_u=1,\dots,N_u$, as given by \eqref{eq-gain}. To achieve this, we first need to derive the marginal distribution of the channel gain $g_u^{n_u}=\left|H_u^{n_u}\right|^2$. It was proven in \cite{Badiu2020Communication} that for sufficiently large $K$, the channel coefficient $H_u^{n_u}$ has a non-circularly symmetric complex normal RV with independent real and imaginary parts $U_u=\mathfrak{R}\left(H_u^{n_u}\right)$ and $V_u=\mathfrak{I}\left(H_u^{n_u}\right)$, respectively, and $U\sim\mathcal{N}\left(\mu_U,\sigma^2_U\right)$ and  $V\sim\mathcal{N}\left(0,\sigma^2_V\right)$ with the parameters $\mu_U=\varphi_1\tilde{a}^2$, $\sigma^2_U=\frac{1}{2K}\left(1+\varphi_2-2\varphi_1^2\tilde{a}^4\right)$, and $\sigma^2_V=\frac{1}{2K}\left(1-\varphi_2\right)$. Also, for large $K$ and $\varphi_1>0$, $\left|H_u^{n_u}\right|$ is well approximated by a Nakagami-$m$ distribution, and thus the channel gain $g_u^{n_u}$ follows a Gamma distribution with the following PDF and CDF \cite{Badiu2020Communication}
\begin{align} \label{eq-g-pdf}
f_{g_u^{n_u}}\left(g\right) = \frac{m^m}{\Gamma\left(m\right)\overline{g}_u^m} g^{m-1}\exp\left(-\frac{mg}{\overline{g}_u}\right)
\end{align}
and
\begin{align}\label{eq-g-cdf}
F_{g_u^{n_u}}\left(g\right)= \frac{1}{\Gamma\left(m\right)}\Upsilon\left(m,\frac{mg}{\overline{g}_u}\right),
\end{align}
where $\Upsilon\left(\cdot\right)$ denotes the lower incomplete Gamma function, and $\overline{g}_u$ and $m$ represent the average channel gain and the fading parameter, respectively, which are defined as
\begin{align}\label{eq-g}
\overline{g}_u=\beta_u^2\varphi_1^2\tilde{a}^4
\end{align}
and
\begin{align}\label{eq-m}
m=\frac{K}{2}\frac{\varphi_1^2\tilde{a}^4}{1+\varphi_2-2\varphi_1^2\tilde{a}^4}.
\end{align}

Following \cite[Fig.~1]{Badiu2020Communication}, we assume all channels to exhibit Rician fading with unit power and Rice factor $\mathcal{K}$, which gives 
\begin{align}
a_\mathrm{b}=a_u=\sqrt{\frac{\pi}{4\left(\mathcal{K}+1\right)}}{}_1\mathcal{F}_1\left(-\frac{1}{2};1;\mathcal{K}\right). 
\end{align}
Regarding \eqref{eq-gain}, the CDF of the equivalent channel gain at the FAS-equipped user can be mathematically expressed as
\begin{align}
F_{g_{\mathrm{fas},u}}\left(g\right) &= \Pr\left(\max\left\{g_u^1,\dots,g_u^{n_u}\leq g\right\}\right)\\
&= F_{g_u^1,\dots,g_u^{n_u}}\left(g,\dots,g\right), \label{eq-joint}
\end{align}
in which \eqref{eq-joint} represents the joint multivariate CDF of $N_u$ correlated channel gains $g_u^{n_u}$, where $g_u^{n_u}$ follows the Gamma distribution as revealed in \eqref{eq-g-pdf} and \eqref{eq-g-cdf}. Following \cite{Rostami2021copula}, we exploit Sklar's theorem for deriving the distribution of $g_{\mathrm{fas},u}$, which can generate the joint multivariate distribution of two or more correlated RVs by only knowing the marginal distribution. Hence, we have
\begin{align}
F_{g_{\mathrm{fas},u}}\left(g\right)=C\left(F_{g_u^1}\left(g\right),\dots,F_{g_u^{N_u}}\left(g\right);\theta\right),\label{eq-cdf-fas}
\end{align}
where $C\left(\cdot\right):\left[0,1\right]^{N_u}\rightarrow\left[0,1\right]$ denotes the copula function and $\theta$ is the dependence parameter that characterizes the dependency between the correlated channel gains. Also, \cite{Rostami2024gaussian} proved that the elliptical copula accurately captures the spatial correlation in FAS. Hence, we use a more general elliptical Student-$t$ copula, which also covers the popular Gaussian copula as a special case, to derive the distribution of $g_{\mathrm{fas},u}$.

\begin{Proposition}
The CDF and PDF of the equivalent channel gain $g_{\mathrm{fas},u}$ for user $u$ considering phase errors in the FAS-aided STAR-RIS RSMA are given by \eqref{eq-cdf} and \eqref{eq-pdf} (see top of next page), where $t_{\nu_u}^{-1}\left(\cdot\right)$ is the inverse CDF (quantile function) of the univariate $t$-distribution having $\nu_u$ degrees of freedom for user $u$,  $T_{\nu_u,\mathbf{\Sigma}_u}\left(\cdot\right)$ represents the CDF of the multivariate $t$-distribution with correlation matrix $\mathbf{\Sigma}_u$ and $\nu_u$ degrees of freedom for user $u$, and $\theta_u\in\left[-1,1\right]$ denotes the dependence parameter of the Student-$t$ copula, which represents the correlation between the $n_u$-th and $\tilde{n}_u$-th ports in $\mathbf{\Sigma}_u$. Moreover, $\Gamma\left(\cdot\right)$ is the Gamma function, $\left|\mathbf{\Sigma}_u\right|$ defines the determinant of $\mathbf{\Sigma}_u$, and
\begin{figure*}
\begin{align}
F_{g_{\mathrm{fas},u}}\left(g\right)= T_{\nu_u,\mathbf{\Sigma}_u}\left(t_{\nu_u}^{-1}\left(\frac{1}{\Gamma\left(m\right)}\Upsilon\left(m,\frac{mg}{\overline{g}_u}\right)\right),\dots,t_{\nu_u}^{-1}\left(\frac{1}{\Gamma\left(m\right)}\Upsilon\left(m,\frac{mg}{\overline{g}_u}\right)\right);\nu_u,\theta_u\right) \label{eq-cdf}
\end{align}
\hrulefill
\begin{align}
f_{g_{\mathrm{fas},u}}\left(g\right) =  \left(\frac{m^m}{\Gamma\left(m\right)\overline{g}_u^m} g^{m-1}\exp\left(-\frac{mg}{\overline{g}_u}\right)\right)^{N_u} \frac{\Gamma\left(\frac{\nu_u+{N_u}}{2}\right)}{\Gamma\left(\frac{\nu_u}{2}\right)\sqrt{\left(\pi\nu_u\right)^{N_u}\left|\mathbf{\Sigma}_u\right|}}\left(1+\frac{1}{\nu_u}\left(\mathbf{t}^{-1}_{\nu_u}\right)^T\mathbf{\Sigma}_u^{-1}\mathbf{t}^{-1}_{\nu_u}\right)^{-\frac{\nu_u+N_u}{2}} \label{eq-pdf}
\end{align}
\hrulefill
\end{figure*}
\begin{align}
\mathbf{t}^{-1}_{\nu_u} = \left[t^{-1}_{\nu_u}\left(\frac{\Upsilon\left(m,\frac{mg}{\overline{g}_u}\right)}{\Gamma\left(m\right)}\right),\dots,t^{-1}_{\nu_u}\left(\frac{\Upsilon\left(m,\frac{mg}{\overline{g}_u}\right)}{\Gamma\left(m\right)}\right)\right].
\end{align}
\end{Proposition} 

\begin{proof}
As in \eqref{eq-cdf-fas}, $F_{g_{\mathrm{fas},u}}\left(g\right)$ is defined based on the marginal CDF and any arbitrary copula. Based on the insights from \cite{Rostami2024gaussian}, which demonstrate that an elliptical copula can approximate Jakes' model well in FAS, we utilize the Student-$t$ copula. This approach accounts for positive and negative correlations, both linear and non-linear, while also capturing heavy-tail dependencies. The Student-$t$ copula is defined as
\begin{align}
C\left(v_1,\dots,v_d\right) = T_{\nu,\mathbf{\Sigma}}\left(t^{-1}_\nu\left(v_1\right),\dots,t^{-1}_\nu\left(v_d\right);\nu,\theta\right), \label{eq-student}
\end{align}
in which $t_{\nu}^{-1}\left(\cdot\right)$ denotes the quantile function of the univariate $t$-distribution with $\nu$ degrees of freedom, while  $T_{\nu,\mathbf{\Sigma}}\left(\cdot\right)$ is the CDF of the multivariate $t$-distribution with correlation matrix $\mathbf{\Sigma}$ and $\nu$ degrees of freedom. Additionally, $\theta\in\left[-1,1\right]$ serves the dependence parameter of Student-$t$ copula, quantifying the correlation between any two arbitrary RVs. The Student-$t$ copula is recognized as a more general form of the elliptical copula, converging to the Gaussian copula as the degrees of freedom $\nu$ increase, i.e., $\nu\rightarrow\infty$. Furthermore, the dependence parameter of elliptical copulas, such as the Gaussian copula, closely approximates the correlation coefficient obtained from Jakes' model, i.e., $\theta_u\approx \varrho_{n_v,\tilde{n}_v}$ \cite{Rostami2024gaussian}. Hence, by considering $v_{n_u}=F_{g_u^{n_u}}\left(g\right)$ and the marginal CDF in \eqref{eq-g-cdf}, \eqref{eq-cdf} is derived.

Now, by definition, the PDF $f_{g_{\mathrm{fas},u}}\left(g\right)$ is the derivative of the  CDF $F_{g_{\mathrm{fas},u}}\left(g\right)$, as defined in \eqref{eq-cdf-fas}. Thus, we have
\begin{align}
f_{g_{\mathrm{fas},u}}\left(g\right) = \frac{\partial^{N_u}}{\underset{N_u}{\underbrace{\partial g\dots \partial g}}}C\left(F_{g_u^1}\left(g\right),\dots,F_{g_u^{N_u}}\left(g\right);\theta_u\right).
\end{align}
Now, using the chain rule, we have \cite{Rostami2021copula}
\begin{align}
f_{g_{\mathrm{fas},u}}\left(g\right) = c\left(F_{g_u^1}\left(g\right),\dots,F_{g_u^{N_u}}\left(g\right);\theta_u\right)\prod_{n_u=1}^{N_u} f_{g_u^{n_u}}\left(g\right),\label{eq-cop-pdf1}
\end{align}
where $c\left(\cdot\right)$ is the copula density function. For the Student-$t$ copula with correlation matrix $\mathbf{\Sigma}_u$ and $\nu_u$ degrees of freedom, the copula density function is defined as
\begin{align}\nonumber
&c\left(F_{g_u^1}\left(g\right),\dots,F_{g_u^{N_u}}\left(g\right);\theta_u\right)=\\
& \frac{\Gamma\left(\frac{\nu_u+{N_u}}{2}\right)}{\Gamma\left(\frac{\nu_u}{2}\right)\sqrt{\left(\pi\nu_u\right)^{N_u}\left|\mathbf{\Sigma}_u\right|}}\left(1+\frac{1}{\nu_u}\left(\mathbf{t}^{-1}_{\nu_u}\right)^T\mathbf{\Sigma}_u^{-1}\mathbf{t}^{-1}_{\nu_u}\right)^{-\frac{\nu_u+N_u}{2}},\label{eq-cop-pdf}
\end{align}
where $\mathbf{t}^{-1}_{\nu} = \left[t^{-1}_{\nu_u}\left(F_{g_u^{n_u}}\left(g\right)\right),\dots,t^{-1}_{\nu_u}\left(F_{g_u^{N_u}}\left(g\right)\right)\right]$ defines the vector of quantiles corresponding to the marginal distributions. Finally, by inserting \eqref{eq-cop-pdf} into \eqref{eq-cop-pdf1}, \eqref{eq-pdf} is obtained. 
\end{proof}

\subsection{OP Analysis}
OP quantifies the likelihood that the received SINR falls below a specified limit, leading to communication failure. In the FAS-assisted STAR-RIS RSMA scheme, each user receives a combination of a shared common message, its own private message, and the private messages of another user. These are decoded using a two-stage process. A communication outage is observed if the SINR for decoding either the common message or the private message does not meet the required thresholds, represented by $\gamma_\mathrm{th,c}^u$ for the common message and $\gamma_\mathrm{th,p}^u$ for the private message. Therefore, the OP for the considered system model is derived in the following proposition. 

\begin{Proposition}
The OP in the FAS-assisted STAR-RIS RSMA  considering phase errors is given by \eqref{eq-out}, where $\overline{g}_u$ and $m$ are given by \eqref{eq-g} and \eqref{eq-m}, respectively (see top of next page), and $\gamma_{\mathrm{th}}^u=\max\left\{\hat{\gamma}_\mathrm{th,c}^u,\hat{\gamma}_\mathrm{th,p}^u\right\}$ determines the corresponding threshold, in which 
\begin{figure*}
\begin{align}
P_{\mathrm{o},u}  = T_{\nu_u,\mathbf{\Sigma}_u}\left(t_{\nu_u}^{-1}\left(\frac{1}{\Gamma\left(m\right)}\Upsilon\left(m,\frac{m\gamma_\mathrm{th}^u}{\overline{g}_u}\right)\right),\dots,t_{\nu_u}^{-1}\left(\frac{1}{\Gamma\left(m\right)}\Upsilon\left(m,\frac{m\gamma_\mathrm{th}^u}{\overline{g}_u}\right)\right);\nu_u,\theta_u\right) \label{eq-out}
\end{align}
\hrulefill
\end{figure*}
\begin{align}
\hat{\gamma}_\mathrm{th,c}^u = \frac{\gamma_{\mathrm{th,c}}^u}{\overline{\gamma}L_uK^2\left(\alpha_\mathrm{c}-\left(\alpha_\mathrm{p,r}+\alpha_\mathrm{p,t}\right)\gamma_\mathrm{th,c}^u\right)}\label{eq-th1}
\end{align}
and
\begin{align}
\hat{\gamma}_\mathrm{th,p}^u = \frac{\gamma_{\mathrm{th,p}}^u}{\overline{\gamma}L_uK^2\left(\alpha_{\mathrm{p},u}-\alpha_{\mathrm{p},\bar{u}}\gamma_\mathrm{th,p}^u\right)}.\label{eq-th2}
\end{align}
\end{Proposition}

\begin{proof}
User $u$ will experience an outage if either the private message or the common message is in outage, i.e., 
\begin{subequations}
\begin{align}
P_{\mathrm{o},u}& = 1-\Pr\left(\gamma_{\mathrm{c},u}>\gamma_\mathrm{th,c}^u,\gamma_{\mathrm{p},u}>\gamma_\mathrm{th,p}^u\right) \label{p-a}\\ \notag
&=1-\Pr\Bigg(\frac{\overline{\gamma}\alpha_\mathrm{c}L_uK^2g_{\mathrm{fas},u}}{\overline{\gamma}\left(\alpha_\mathrm{p,r}+\alpha_\mathrm{p,t}\right)L_uK^2g_{\mathrm{fas},u}+1}>\gamma_\mathrm{th,c}^u,\label{p-b}\\
&\quad\quad\quad\frac{\overline{\gamma}\alpha_{\mathrm{p},u}L_uK^2g_{\mathrm{fas},u}}{\overline{\gamma}\alpha_{\mathrm{p},\bar{u}}L_uK^2g_{\mathrm{fas},u}+1}>\gamma_\mathrm{th,p}^u\Bigg)\\
& = 1-\Pr\left(g_{\mathrm{fas},u}>\hat{\gamma}_\mathrm{th,c}^u,g_{\mathrm{fas},u}>\hat{\gamma}_\mathrm{th,p}^u\right)\\
& = \Pr\left(g_{\mathrm{fas},u}\leq \max\left\{\hat{\gamma}_\mathrm{th,c}^u,\hat{\gamma}_\mathrm{th,p}^u\right\}\right)\\
& = F_{g_{\mathrm{fas},u}}\left(\gamma_\mathrm{th}^u\right),\label{p-e}
\end{align}
\end{subequations}
where $F_{g_{\mathrm{fas},u}}\left(\cdot\right)$ has been defined in \eqref{eq-cdf}. By inserting $\gamma_{\mathrm{th}}^u=\max\left\{\hat{\gamma}_\mathrm{th,c}^u,\hat{\gamma}_\mathrm{th,p}^u\right\}$ into \eqref{eq-cdf}, \eqref{eq-out} is derived. 
\end{proof}

\begin{corollary}
The asymptotic OP in the high-SNR regime for the considered FAS-assisted STAR-RIS RSMA considering the phase errors is given by (\ref{eq-out-asym}), see top of next page.
\begin{figure*}
\begin{align}
P_{\mathrm{o},u}^\infty  \simeq T_{\nu_u,\mathbf{\Sigma}_u}\left(t_{\nu_k}^{-1}\left(\frac{1}{m\Gamma\left(m\right)}\left(\frac{mg}{\overline{g}_u}\right)^m\right),\dots,t_{\nu_k}^{-1}\left(\frac{1}{m\Gamma\left(m\right)}\left(\frac{mg}{\overline{g}_u}\right)^m\right);\nu_u,\theta_u\right) \label{eq-out-asym}
\end{align}
\hrulefill
\end{figure*}
\end{corollary}

\begin{proof}
In the high SNR regime, i.e., $\overline{\gamma}\rightarrow\infty$, the CDF in \eqref{eq-g-cdf} is derived as
\begin{align}
F_{g_u^{n_u}}^\infty\left(g\right)\simeq \frac{1}{m\Gamma\left(m\right)}\left(\frac{mg}{\overline{g}_u}\right)^m.\label{eq-cdf-asym}
\end{align}
By inserting \eqref{eq-cdf-asym} into the OP definition, \eqref{eq-out-asym} is obtained. 
\end{proof}

\begin{remark}\label{remark1}
The RSMA power allocation factors $\alpha_\mathrm{c}$ and $\alpha_{\mathrm{p},u}$, as well as the outage thresholds $\gamma_\mathrm{th,c}^u$ and $\gamma_\mathrm{th,p}^u$ are mutually dependent through the following constraints
\begin{align}
\gamma_{\mathrm{th,c}}^u<\frac{\alpha_\mathrm{c}}{1-\alpha_\mathrm{c}}\label{eq-con1}
\end{align}
and
\begin{align}
\gamma_{\mathrm{th,p}}^u<\frac{\alpha_{\mathrm{p},u}}{1-\alpha_\mathrm{c}-\alpha_{\mathrm{p},u}},\label{eq-con2}
\end{align}
where \eqref{eq-con1} and \eqref{eq-con2} are derived by assuming $\hat{\gamma}_\mathrm{th,c}^u>0$ and $\hat{\gamma}_\mathrm{th,p}^u>0$ in \eqref{eq-th1} and \eqref{eq-th2}, respectively. If the constraints are not satisfied, the system operates in an invalid region. This behavior is thoroughly illustrated in Fig.~\ref{fig_op_alphac} in Section \ref{sec-num}.  
\end{remark}

\subsection{AC Analysis}
The channel capacity for both the common and private messages $v\in\left\{\mathrm{c,p}\right\}$ of user $u$ can be found by
\begin{align}
\mathcal{C}_{v,u} = \log_2\left(1+\gamma_{v,u}\right), 
\end{align}
and thus, the corresponding AC (bps/Hz) is expressed as
\begin{align}\label{eq-c-def}
\overline{\mathcal{C}}_{v,u} =\mathbb{E}\left[\mathcal{C}_{v,u}\right]= \int_0^\infty \log_2\left(1+\gamma_{v,u}\left(g\right)\right)f_{g_{\mathrm{fas},u}}\left(g\right)\mathrm{d}g,
\end{align}
where $f_{g_{\mathrm{fas},u}}\left(g\right)$ is provided by \eqref{eq-g-pdf}. Given the form of \eqref{eq-g-pdf}, it is clear that obtaining a closed-form expression for \eqref{eq-c-def} is mathematically challenging. While numerical techniques, such as Gauss-Laguerre quadrature, could offer practical solutions, we instead introduce a heuristic method to estimate the AC, as presented in the following proposition.

\begin{Proposition}
The ACs in the FAS-assisted STAR-RIS RSMA  considering phase errors for the common and private messages are, respectively, given by
\begin{align}
\overline{\mathcal{C}}_{\mathrm{c},u}\approx \log_2\left(1+ \frac{\overline{\gamma}\alpha_\mathrm{c}L_uK^2\overline{g}_{\mathrm{fas},u}}{\overline{\gamma}\left(\alpha_\mathrm{p,r}+\alpha_\mathrm{p,t}\right)L_uK^2\overline{g}_{\mathrm{fas},u}+1}\right)
\end{align}
and
\begin{align}
\overline{\mathcal{C}}_{\mathrm{p},u}\approx \log_2\left(1+ \frac{\overline{\gamma}\alpha_{\mathrm{p},u}L_uK^2\overline{g}_{\mathrm{fas},u}}{\overline{\gamma}\alpha_{\mathrm{p},\bar{u}}L_uK^2\overline{g}_{\mathrm{fas},u}+1}\right),
\end{align}
where
\begin{align}
\overline{g}_{{\mathrm{fas},u}}\approx \overline{g}_u+\frac{\overline{g}_u^2}{m}\sqrt{1+\rho_\mathrm{eff}}\Phi^{-1}\left(\frac{N_u}{N_u-1}\right),\label{eq-mean-corr}
\end{align}
in which $\rho_\mathrm{eff}=\frac{\varrho_{n_u,\tilde{n}_u}}{N_u\left(N_u-1\right)}$ and $\Phi^{-1}\left(\cdot\right)$ is the standard normal quantile function.
\end{Proposition}

\begin{proof}
By applying Jensen's inequality into \eqref{eq-c-def}, we have 
\begin{align}
\overline{\mathcal{C}}_{v,u} \approx \log_2\left(1+\mathbb{E}\left[\gamma_{v,u}\left(g\right)\right]\right), 
\end{align}
and as a result, the AC expressions for the messages $v\in\left\{\mathrm{c,p}\right\}$ can be derived as
\begin{align}
\overline{\mathcal{C}}_{\mathrm{c},u}\approx \log_2\left(1+ \frac{\overline{\gamma}\alpha_\mathrm{c}L_uK^2\mathbb{E}\left[g_{\mathrm{fas},u}\right]}{\overline{\gamma}\left(\alpha_\mathrm{p,r}+\alpha_\mathrm{p,t}\right)L_uK^2\mathbb{E}\left[g_{\mathrm{fas},u}\right]+1}\right)
\end{align}
and
\begin{align}
\overline{\mathcal{C}}_{\mathrm{p},u}\approx \log_2\left(1+ \frac{\overline{\gamma}\alpha_{\mathrm{p},u}L_uK^2\mathbb{E}\left[g_{\mathrm{fas},u}\right]}{\overline{\gamma}\alpha_{\mathrm{p},\bar{u}}L_uK^2\mathbb{E}\left[g_{\mathrm{fas},u}\right]+1}\right). 
\end{align}
Now, we need to find the expectation of $g_{\mathrm{fas},u}$, which is the expectation of the maximum of $N_u$ correlated Gamma RVs, as defined in \eqref{eq-gain}. To achieve this, we propose a heuristic approach in which we first calculate the expectation of the maximum of $N_u$ independent Gamma RVs. This is then extended to the correlated case by incorporating a heuristic term based on the Student-$t$ copula. Hence, assuming that $g_u^1,\dots,g_u^{n_u}$ for $n_u=1,\dots,N_u$ are i.i.d.~RVs, the distribution of the maximum $g_{\mathrm{ind},u} = \left\{g_u^1,\dots,g_u^{n_u}\right\}$ is given by $F_{g_{\mathrm{ind},u}}\left(g\right)=F_{g_u^{n_u}}\left(g\right)^{N_u}$, where $F_{g_u^{n_u}}\left(g\right)$ is the CDF of each individual RV. 

As $N_u$ grows large, the maximum becomes a more predictable extreme value. In the large $N_u$ limit, the distribution of $g_\mathrm{ind}$ starts to resemble an extreme value distribution, with $F_{g_{\mathrm{ind},u}}(g)$ approaching $1$ for most values of $g$. This reflects the increasing dominance of the largest values in the sample as $N_u$ increases. For large $N_u$, $\mathbb{E}\left[g_{\mathrm{ind},u}\right]$ is approximately the quantile of the individual CDF evaluated at the extreme upper probability $\frac{N_u-1}{N_u}$. This quantile corresponds to the threshold beyond which the maximum value is found with very high confidence. Specifically, as $N_u$ increases, the maximum converges to a value $g^*$ such that the probability of $g_{\mathrm{ind},u}$ being less than $g^*$ is approximately \cite{Leadbetter2012extremes}
\begin{align}
F_{g_{\mathrm{ind},u}}\left(g^*\right) = \frac{N_u-1}{N_u}.
\end{align}
Substituting $F_{g_{\mathrm{ind},u}}\left(g\right)=F_{g_u^{n_u}}(g)^{N_u}$ and taking the $N_u$-th root, we have
\begin{align}
F_{g_u^{n_u}}\left(g^*\right) = \left(\frac{N_u-1}{N_u}\right)^{1/N_u}\overset{(a)}{\approx} \frac{N_u-1}{N_u},
\end{align}
where $(a)$ is derived by using a first-order approximation of $\left(\frac{N_u-1}{N_u}\right)^{1/N_u}$ when $N_u$ is large \cite{Abramowitz1972applied}. Therefore, the quantile $g^*$ can be expressed using the inverse CDF as
\begin{align}
g^*\approx F^{-1}_{g_u^{n_u}}\left(\frac{N_u-1}{N_u}\right). \label{eq-g*}
\end{align}
For large $N_u$, the maximum becomes increasingly predictable as its distribution converges to an extreme value distribution, specifically the Gumbel distribution for light-tailed variables like the Gamma distribution. This theoretical foundation justifies the use of the quantile $F^{-1}_{g_u^{n_u}}\left(\frac{N_u-1}{N_u}\right)$ to approximate the mean of the maximum. It reflects the high probability of the largest observed value being close to this threshold.

Now, for the Gamma distribution, which has finite moments, we can approximate the CDF $F_{g_u^{n_u}}\left(g\right)$ around its mean $\mu=\overline{g}_u$ using the standard normal CDF $\Phi\left(\frac{g-\mu}{\sigma}\right)$ \cite{Johnson1995continuous,Kendall1987kendall}, i.e., 
\begin{align}
F_{g_u^{n_u}}\left(g\right)\approx\Phi\left(\frac{g-\mu}{\sigma}\right),
\end{align}
where $\sigma=\sqrt{\mathrm{Var}\left[g\right]}=\frac{\overline{g}_u^2}{m}$. Subsequently, the inverse CDF is derived as \cite{David2004order}
\begin{align}
F^{-1}_{g_u^{n_u}}\left(g\right)\approx\mu+\sigma\Phi^{-1}\left(p\right),
\end{align}
where by substituting $p=\frac{N_u-1}{N_u}$, \eqref{eq-g*} is rewritten as
\begin{align}
g^* \approx \overline{g}_u+\frac{\overline{g}_u^2}{m}\Phi^{-1}\left(\frac{N_u-1}{N_u}\right).
\end{align}
For large $N_u$, the mean of the maximum $g_{\mathrm{ind},u}$ of $N_u$ RVs is close to $g^*$ \cite{Mardia2000directional}, giving
\begin{align}
\mathbb{E}\left[g_{\mathrm{ind},u}\right]\approx \overline{g}_u+\frac{\overline{g}_u^2}{m}\Phi^{-1}\left(\frac{N_u-1}{N_u}\right). \label{eq-mean-ind}
\end{align}
Now we incorporate dependence using the Student-$t$ copula with the correlation matrix $\mathbf{\Sigma}_u$ that governs the dependence between the RVs. It is noted that while the full correlation matrix $\mathbf{\Sigma}_u$ provides a detailed pairwise dependency structure, working directly with it for large $N_u$ becomes computationally prohibitive. Therefore, given that copula-induced dependency influences the effective spread of the maximum, we approximate the effect of correlation by scaling the independent approximation with a dependency factor $\rho_\mathrm{eff}$, the average pairwise correlation coefficient extracted from $\mathbf{\Sigma}_u$ (excluding the diagonal entries, which are all ones), as
\begin{align}
\rho_\mathrm{eff} = \frac{\mathop{\mathbf{\Sigma}_u} \varrho_{n_u, \tilde{n}_u}}{N_u \left(N_u - 1\right)}, 
\end{align}
where $N_u\left(N_u-1\right)$ is the number of off-diagonal terms in the correlation matrix. Therefore, by summing these terms and dividing by the total number of pairs is derived, which ensures that $\rho_\mathrm{eff}$ is normalized and lies in the range $\left[0,1\right]$. Consequently, given that the variance of the maximum of correlated RVs differs from the independent case, the adjusted variance can be expressed as
\begin{align}
\mathrm{Var}\left[g_{\mathrm{fas},u}\right]\approx \mathrm{Var}\left[g_{\mathrm{ind},u}\right]\left(1+\rho_\mathrm{eff}\right).\label{eq-varcor}
\end{align}
By applying \eqref{eq-varcor} into \eqref{eq-mean-ind}, $\overline{g}_{{\mathrm{fas},u}}$ is obtained as \eqref{eq-mean-corr}.
\end{proof}

\begin{remark}\label{remark2}
The heuristic method employed to approximate the expectation of the maximum of $N_u$ correlated Gamma RVs, while based on simplifying assumptions, provides useful insights into the AC's behavior, particularly for moderate correlations. By leveraging the effective correlation coefficient derived from the Student-$t$ copula, this approach captures the main characteristics of the maximum, offering a computationally efficient approximation. Although the method might not fully account for complex dependencies due to combining the individual distribution's mean and variance with the effective correlation parameter $\rho_\mathrm{eff}$, it serves as a practical tool when the correlation is moderate, as it balances accuracy with simplicity. As will be demonstrated in Section \ref{sec-num}, simulations show that this approximation performs well across various scenarios, validating its effectiveness for analyzing correlated RVs in the context of the model under consideration.
\end{remark}

\section{Numerical Results}\label{sec-num}
In this section, we evaluate the performance of the considered FAS-aided  STAR-RIS RSMA system in terms of the OP and AC. To do so, we have set the system parameters to $\chi=2.1$, $\alpha_\mathrm{c}=0.6$, $\nu=40$, $\mathcal{K}=1$, $\kappa=8$, $N_u=4$, $W_u=0.5\lambda^2$, $\gamma_\mathrm{th,c}^u=0$ dB, $\gamma_\mathrm{th,p}^\mathrm{r}=0$ dB, and $\gamma_\mathrm{th,p}^\mathrm{t}=-7$ dB. Additionally, for the considered ES mode, we set $\beta_{\mathrm{r}}=0.8$ and $\beta_{\mathrm{t}}=\sqrt{1-\beta_{\mathrm{r}}^2}$. Furthermore, we assume that the BS is located at the origin $(0,0,0)$, the STAR-RIS is located at $(40,40,0)$, and the FAS-equipped users $\mathrm{r}$ and $\mathrm{t}$ are placed at $(20,20,0)$m and $(20,20,0)$m, respectively. We also set the power allocation factors of the private messages as $\alpha_\mathrm{p,r}=0.75\left(1-\alpha_\mathrm{c}\right)$ and $\alpha_\mathrm{p,t}=0.25\left(1-\alpha_\mathrm{c}\right)$. We use the MATLAB function \texttt{copulacdf} to implement the Student-$t$ copula, which is expressed in terms of the joint CDF of the multivariate $t$-distribution. We also consider the ideal phase case for both the FAS and TAS under RSMA as a benchmark. 

\begin{figure}[!t]
\centering
\includegraphics[width=1\columnwidth]{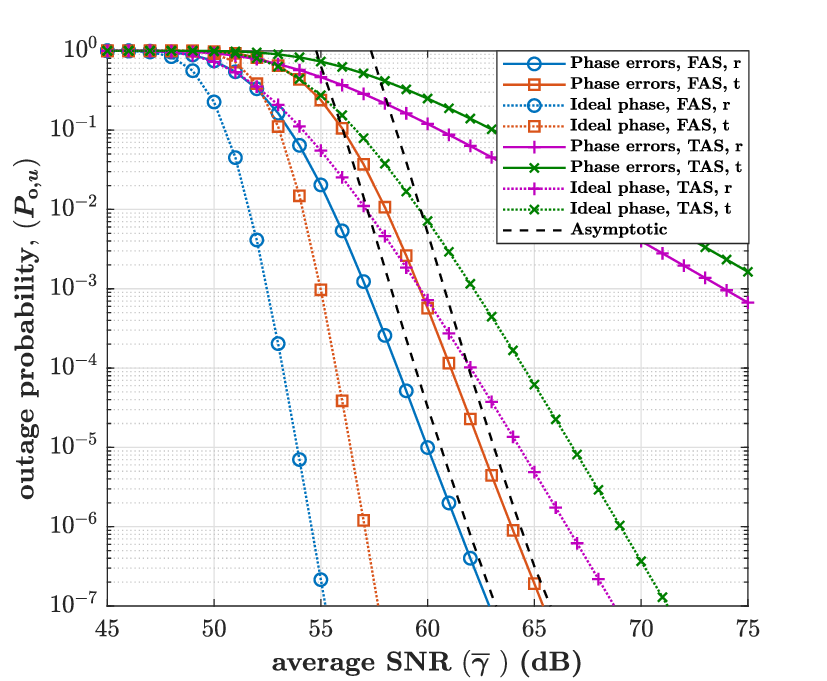}
\caption{The OP results of FAS-assisted STAR-RIS RSMA versus the average SNR $\overline{\gamma}$ with phase errors and ideal phases for $K=30$.}\label{fig_op_snr_phase}
\end{figure}

Fig.~\ref{fig_op_snr_phase} shows the behavior of OP against the average SNR $\overline{\gamma}$ for the RSMA user $u$ in the proposed STAR-RIS scenario when both phase errors and ideal phase cases are considered. First, it can be observed that the asymptotic results closely align with the OP curves in the high-SNR regime, confirming the accuracy of our theoretical analysis. Furthermore, we can see that as $\overline{\gamma}$ increases, the OP performance improves for both FAS and TAS. This behavior is expected, as a higher average SNR reflects enhanced channel conditions, leading to stronger received signals and reduced likelihood of outage. It is also evident that the ideal phase configuration provides a lower OP for both FAS and TAS compared to the phase errors scenario. This disparity is primarily due to the superior phase alignment in the ideal phase case, which ensures beamforming efficiency, and enhanced diversity gain. On the other hand, the phase errors  disrupt the alignment of reflected and transmitted signals, leading to degraded performance and higher OP. 

\begin{figure}[!t]
\centering
\includegraphics[width=1\columnwidth]{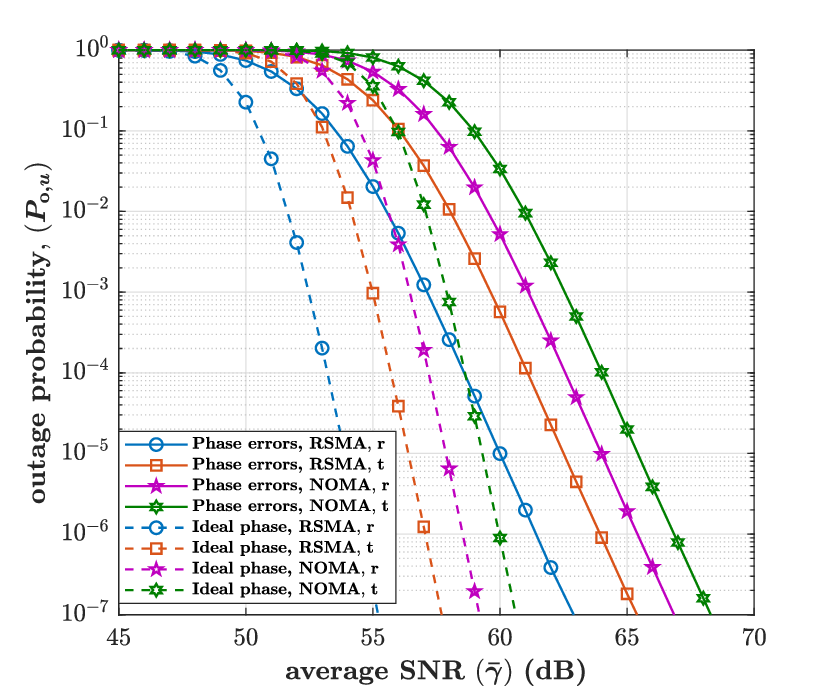}
\caption{The OP results of FAS-assisted STAR-RIS versus the average SNR $\overline{\gamma}$ with phase errors and ideal phase for different multiple access scenarios when $K=30$.}\label{fig_op_noma}
\end{figure}

In Fig.~\ref{fig_op_noma}, we examine the OP performance of the STAR-RIS-aided system by comparing RSMA to NOMA. The results demonstrate that RSMA outperforms NOMA across different conditions, including both phase errors and ideal phase scenarios. Specifically, RSMA achieves a significantly lower OP at higher SNR levels, highlighting its superior interference management and more flexible resource allocation. Unlike NOMA, which relies on rigid power domain multiplexing, RSMA dynamically splits messages into common and private parts, enabling more efficient interference mitigation and improved decoding reliability. This adaptability allows RSMA to better handle variations in channel conditions, making it particularly advantageous in dynamic wireless environments. Additionally, RSMA maintains a consistent performance advantage over NOMA for both users, even in the presence of phase errors, demonstrating its robustness against channel imperfections.

\begin{figure}[!t]
\centering
\includegraphics[width=1\columnwidth]{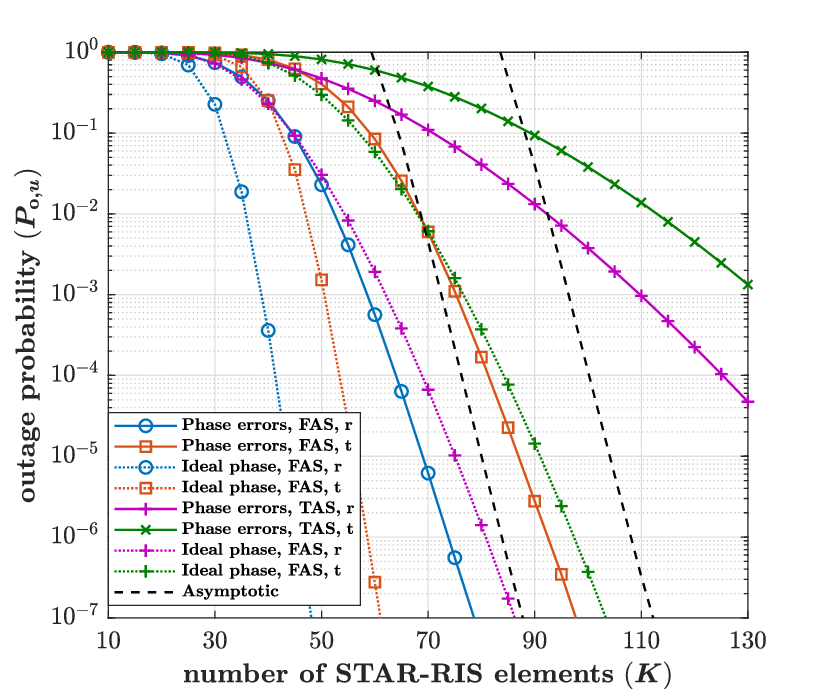}
\caption{The OP results of FAS-assisted STAR-RIS RSMA versus the number of STAR-RIS elements $K$  with phase errors and ideal phase for $\overline{\gamma}=50$ dB.}\label{fig_op_k}
\end{figure}

The impact of the number of STAR-RIS elements $K$ on the OP performance is presented in Fig.~\ref{fig_op_k} for the both FAS and TAS when the phase errors and ideal phase are considered. It can be seen that as $K$ increases, the system experiences a significant reduction in OP due to the enhanced array gain and more focused beamforming capabilities. The larger aperture created by additional elements allows for better spatial control over signal propagation, improved constructive interference, and more effective mitigation of fading effects, resulting in superior link reliability. By closely examining the curves in Fig.~\ref{fig_op_k}, it is also evident that the performance gap between the ideal phase and phase errors scenarios widens as the number of STAR-RIS elements increases. This is mainly because, in the ideal phase case, each additional element perfectly contributes to constructive interference, maximizing the array gain and beamforming efficiency. However, in the phase errors scenario, phase misalignments introduced by errors accumulate as the number of elements grows, leading to greater degradation in performance. As a result, while both cases benefit from more elements, the relative advantage of ideal phase alignment becomes increasingly evident as the array size increases.

\begin{figure}[!t]
\centering
\includegraphics[width=1\columnwidth]{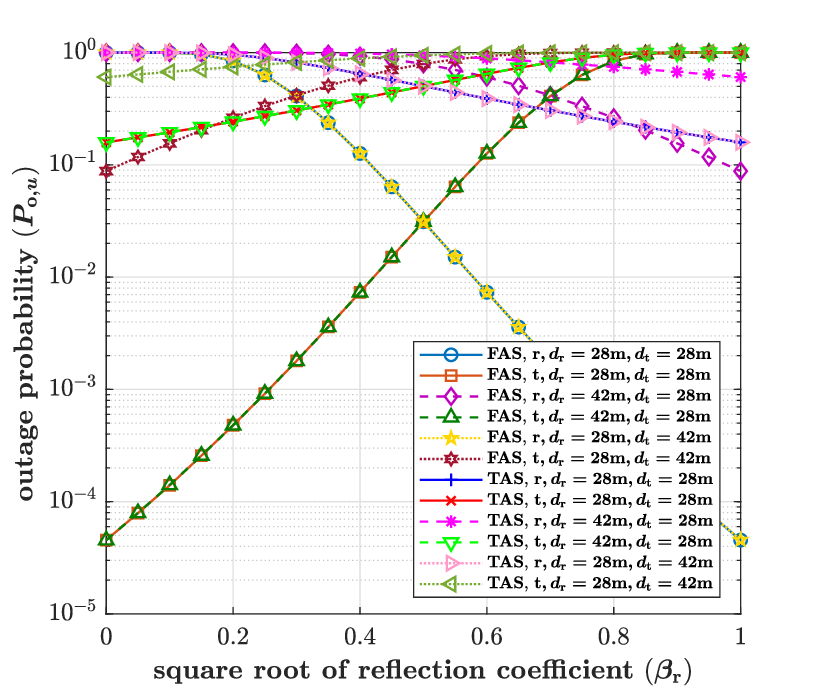}
\caption{The OP results of FAS-assisted STAR-RIS RSMA versus the square root of reflection coefficient $\beta_\mathrm{r}$ with phase errors for $\overline{\gamma}=50$ dB and $K=55$.}\label{fig_op_beta}
\end{figure}

The effect of the square root of transmission coefficient $\beta_\mathrm{r}=\sqrt{1-\beta_\mathrm{t}^2}$ on the OP performance for different values of distances under the phase errors case is indicated in Fig.~\ref{fig_op_beta}. For both FAS and TAS, we can see that as $\beta_\mathrm{r}$ grows, the OP for user $\mathrm{r}$ and user $\mathrm{t}$ decreases and increases, respectively. This behavior is mainly attributed to the trade-off between power allocation to the reflection and transmission paths. Specifically, when $\beta_\mathrm{r}$ increases, a larger portion of the incident signal is reflected towards user $\mathrm{r}$, resulting in a stronger received signal at this user. This enhances the SNR for user $\mathrm{r}$, thereby reducing the likelihood of an outage and improving its OP performance. On the other hand, the increase in $\beta_\mathrm{r}$ directly reduces the power allocated to the transmission path, decreasing the signal strength at user $\mathrm{t}$. Consequently, the SNR for user $\mathrm{t}$ drops, leading to a higher probability of outage and hence a higher OP for user $\mathrm{t}$. Furthermore, the distance between the STAR-RIS elements and each user plays a crucial role in determining the required reflection and transmission power. As the distance between the STAR-RIS and a user grows, the effective path gain decreases, which leads to a reduction in SNR. To compensate for this reduction and maintain the same OP, a higher value of $\beta_u$ is required. For example, if user $\mathrm{r}$ (user $\mathrm{t}$) is placed farther away, the system must allocate more power to the reflection (transmission) path to compensate for the weaker signal at that distance. Therefore, the impact of $\beta_\mathrm{r}$ ($\beta_\mathrm{t}$) becomes more significant at larger distances, as the system relies more on the reflection (transmission) path to support user $\mathrm{r}$ (user $\mathrm{t}$).

\begin{figure}[!t]
\centering
\includegraphics[width=1\columnwidth]{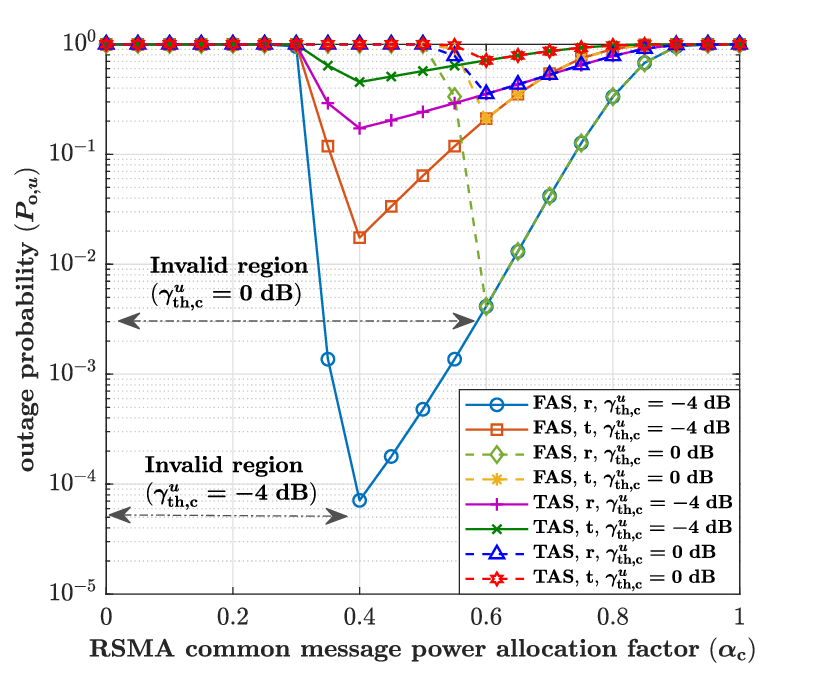}
\caption{The OP results of FAS-assisted STAR-RIS RSMA versus the common message power allocation factor $\alpha_\mathrm{c}$ with phase errors for $\overline{\gamma}=50$ dB and $K=55$.}\label{fig_op_alphac}
\end{figure}

The OP performance versus the RSMA common message power allocation factor $\alpha_\mathrm{c}$ for different values of the SNR threshold $\gamma_\mathrm{th,c}^u$ is presented in Fig.~\ref{fig_op_alphac}. As discussed in Remark \ref{remark1}, $\alpha_\mathrm{c}$ and $\gamma_\mathrm{th,c}^u$ must satisfy condition \eqref{eq-con1}; otherwise, they fall into an invalid region, as depicted in the figure. We see that the OP initially decreases with increasing $\alpha_\mathrm{c}$, reaches an optimal value, and then increases, forming a convex-like shape. The optimal $\alpha_\mathrm{c}$ depends on $\gamma_\mathrm{th,c}^u$ but remains independent of the user type. Accordingly, the OP weakens as $\gamma_\mathrm{th,c}^u$  deteriorates, while the invalid region becomes wider, with its upper limit shifting towards the right. This convex-like behavior is attributed to the trade-off regarding the power allocation over the common and private messages. For higher values of $\alpha_\mathrm{c}$, $\alpha_\mathrm{p}$ becomes insufficient, causing the private message SNR to fall below the threshold and leading to outages. Conversely, for lower values of $\alpha_\mathrm{c}$, the common message SNR cannot meet the threshold, resulting in an outage of the common message. Thus, the system experiences outages at both extremes of $\alpha_\mathrm{c}$, so it is important to find a balance in the power allocation.

\begin{figure}[!t]
\centering
\includegraphics[width=1\columnwidth]{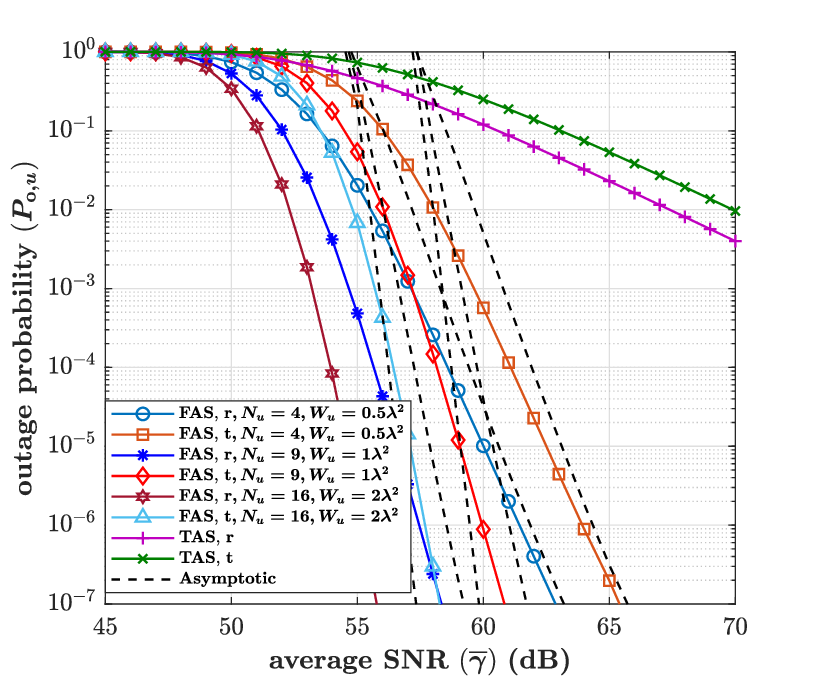}
\caption{The OP results of FAS-assisted STAR-RIS RSMA versus the average SNR $\overline{\gamma}$ with phase errors for different FAS sizes $W_u$ and numbers of FAS elements $N_u$, considering $K=30$.}\label{fig_op_nw}
\end{figure}

Fig.~\ref{fig_op_nw} demonstrates how the FAS size $W_u$ and the number of FAS nodes $N_u$ influence the OP performance as $\overline{\gamma}$ varies. From the figure, it is evident that simultaneously increasing both $W_u$ and $N_u$ results in a lower OP. This improvement is primarily due to the larger fluid antenna size, which enhances spatial separation between the ports, thereby reducing spatial correlation. Although increasing the number of ports can lead to stronger spatial correlation, it also has the potential to improve channel capacity, diversity gain, and spatial multiplexing. Therefore, by simultaneously increasing $W_u$ and $N_u$, the spatial correlation between the fluid antenna ports becomes balanced and hence, a lower OP is achieved. Additionally, it is observed that TAS results in a higher OP compared to the deployment of FAS in RSMA users. This contrast underscores the limitations of fixed antenna configurations in effectively utilizing spatial diversity, especially in channels with rich scattering.

\begin{figure}[!t]
\centering
\includegraphics[width=1\columnwidth]{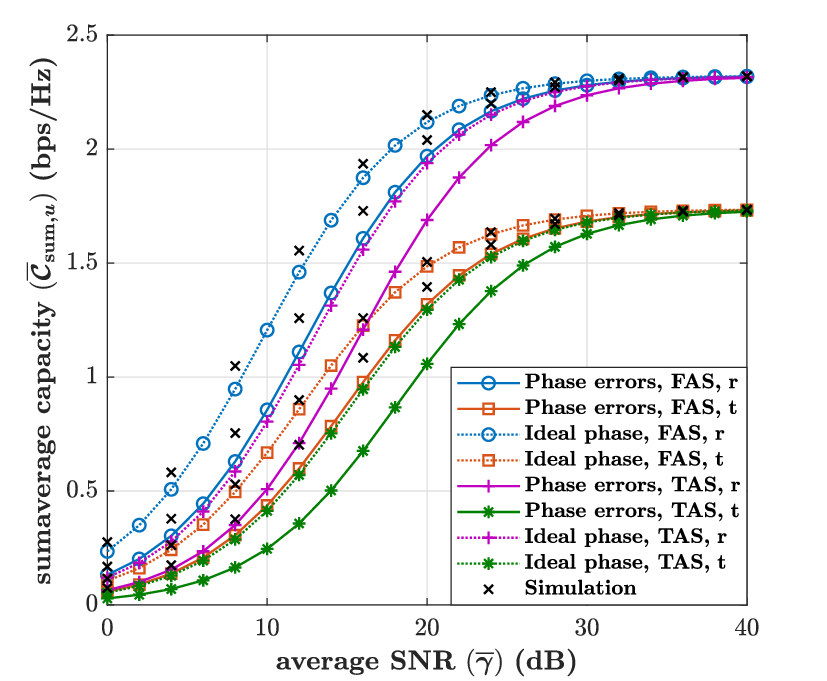}
\caption{The sum AC results of FAS-assisted STAR-RIS RSMA versus the average SNR $\overline{\gamma}$ with phase errors and ideal phases for $K=30$.}\label{fig_ac_snr_phase}
\end{figure}

The results in Fig.~\ref{fig_ac_snr_phase} study the impact of phase errors on the sum AC, $\overline{\mathcal{C}}_{\mathrm{sum},u}=\overline{\mathcal{C}}_{\mathrm{c},u}+\overline{\mathcal{C}}_{\mathrm{p},u}$. First of all, as expected, a small gap between the simulation results and the theoretical predictions is observed. This discrepancy arises from the use of a heuristic technique to approximate the expectation of the maximum of $N_u$ correlated Gamma RVs, as discussed in Remark \ref{remark2}. The heuristic method, based on simplifying assumptions, provides a computationally efficient approximation that captures the essential behavior of the system, especially in cases involving moderate correlations. Despite the gap, the results indicate that this approach performs well across a range of scenarios, making it a useful approximation for evaluating the AC performance under correlated conditions. For both FAS and TAS, we see that as $\overline{\gamma}$ increases, the sum AC initially improves but eventually saturates. This occurs because in the RSMA scheme, rate splitting effectively manages interference and optimizes power allocation at moderate SNR levels. As $\overline{\gamma}$ grows, the SINR expressions in \eqref{eq-snr-c} and \eqref{eq_SINR_p} converge to deterministic values, depending on the power allocation factors for the common and private streams. Beyond this threshold, further increases in SNR do not significantly improve the sum AC, resulting in saturation. Additionally, the case with ideal phase results in a higher level of sum AC compared to the scenario with phase errors. This is because with ideal phase, STAR-RIS can perfectly align the transmitted signals with the desired reflection coefficients, thereby maximizing the signal power at the receiver and minimizing interference. In contrast, phase errors introduce misalignment between the transmitter and STAR-RIS, leading to suboptimal signal reflection and increased interference. As a result, the effective channel gains are reduced in the presence of phase errors, causing a degradation in the sum AC. Noteworthy, the asymptotic behavior for the FAS and TAS schemes is independent of the phase errors, so does the performance degradation because the phase mismatch is reduced as $\overline\gamma$ grows.

\begin{figure}[!t]
\centering
\includegraphics[width=1\columnwidth]{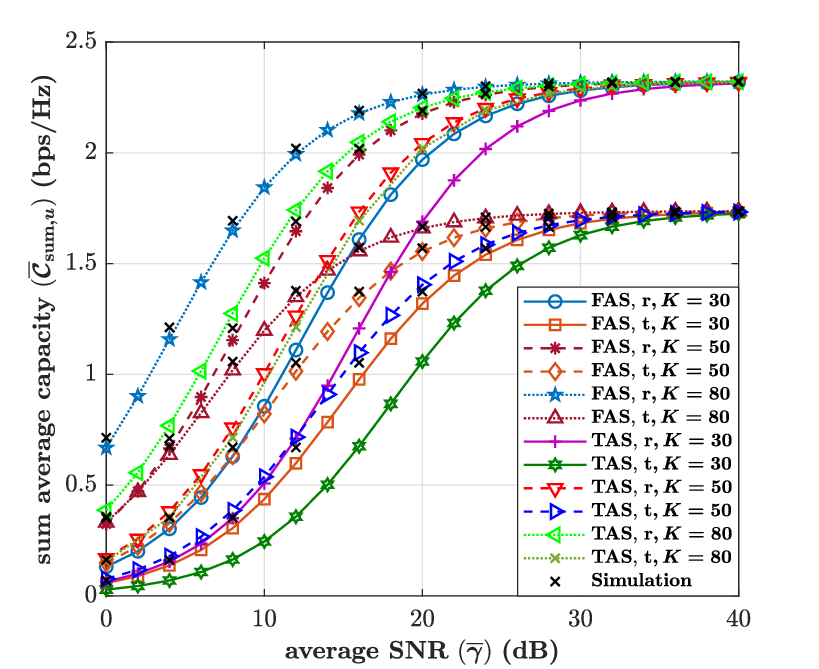}
\caption{The sum AC results of FAS-assisted STAR-RIS RSMA versus the average SNR $\overline{\gamma}$ with phase errors for different numbers of STAR-RIS elements $K$.}\label{Fig_ac_k}
\end{figure}

Finally, Fig.~\ref{Fig_ac_k} studies how the number of STAR-RIS elements, $K$, affects the sum AC. As we can see, a larger $K$ provides a higher sum AC for both FAS and TAS. Increasing $K$ also results in a higher sum AC for both FAS and TAS. This gain occurs because a larger number of STAR-RIS elements increases the surface's ability to adaptively reconfigure the channel, allowing for better alignment of the reflected signals. With more elements, STAR-RIS can have more precise control over the spatial distribution of the reflected signals, effectively enhancing the SNR and reducing interference. Furthermore, a larger $K$ enables better exploitation of spatial diversity, thereby improving the capacity for both RSMA users. As such, the increase in the number of STAR-RIS elements contributes to a more efficient channel utilization, leading to a higher sum AC. We can also see that FAS provides a higher sum AC than the TAS counterpart. This is because FAS can dynamically adjust its configuration, improving channel alignment and spatial diversity, while TAS, with its fixed setup, lacks this flexibility, limiting its performance in varying conditions. As in the previous figure, the sum AC saturates to a value that only depends on the RSMA power allocation factors, as the dependence on $K$ is cancelled out for sufficiently large $\overline{\gamma}$.

\section{Conclusion}\label{sec-con}
This paper presented a comprehensive investigation of FAS-assisted STAR-RIS communication systems employing RSMA signaling, with a particular focus on the impact of phase errors under the ES protocol. By modeling the phase errors with a generic distribution and deriving the equivalent channel gain characterized by the multivariate $t$-distribution, we provided a realistic and robust analysis of the system performance. Compact analytical expressions for key performance metrics, such as outage OP and AC were derived, with the latter being provided by a novel heuristic approach. Our numerical results demonstrated that integrating FAS into STAR-RIS-enabled RSMA communication systems yields significant performance improvements compared to TAS, particularly for mitigating interference and enhancing adaptability. Additionally, we evaluated the impact of phase errors in the STAR-RIS, revealing the trade-offs associated with practical scenarios. 

%\bibliographystyle{IEEEtran}
%\bibliography{refs.bib}

\end{document}